\documentclass[a4paper,english,aps,pra,twocolumn,showpacs,superscriptaddress,superscriptaddress,amsthm,amsmath,amssymb]{revtex4}
\usepackage[utf8]{inputenc}
\usepackage{babel}
\usepackage{refstyle}
\usepackage{amsthm}
\usepackage{amssymb}
\usepackage[unicode=true,pdfusetitle,
 bookmarks=true,bookmarksnumbered=false,bookmarksopen=false,
 breaklinks=false,pdfborder={0 0 1},backref=false,colorlinks=false]
 {hyperref}

\makeatletter

\AtBeginDocument{\providecommand\lemref[1]{\ref{lem:#1}}}
\AtBeginDocument{\providecommand\thmref[1]{\ref{thm:#1}}}
\AtBeginDocument{\providecommand\propref[1]{\ref{prop:#1}}}
\pdfpageheight\paperheight
\pdfpagewidth\paperwidth

\RS@ifundefined{subref}
  {\def\RSsubtxt{section~}\newref{sub}{name = \RSsubtxt}}
  {}
\RS@ifundefined{thmref}
  {\def\RSthmtxt{theorem~}\newref{thm}{name = \RSthmtxt}}
  {}
\RS@ifundefined{lemref}
  {\def\RSlemtxt{lemma~}\newref{lem}{name = \RSlemtxt}}
  {}

\@ifundefined{textcolor}{}
{%
 \definecolor{BLACK}{gray}{0}
 \definecolor{WHITE}{gray}{1}
 \definecolor{RED}{rgb}{1,0,0}
 \definecolor{GREEN}{rgb}{0,1,0}
 \definecolor{BLUE}{rgb}{0,0,1}
 \definecolor{CYAN}{cmyk}{1,0,0,0}
 \definecolor{MAGENTA}{cmyk}{0,1,0,0}
 \definecolor{YELLOW}{cmyk}{0,0,1,0}
}
  \theoremstyle{plain}
  \newtheorem{thm}{\protect\theoremname}
  \theoremstyle{plain}
  \newtheorem{lem}{\protect\lemmaname}
 \ifx\proof\undefined\
   \newenvironment{proof}[1][\proofname]{\par
     \normalfont\topsep6\p@\@plus6\p@\relax
     \trivlist
     \itemindent\parindent
     \item[\hskip\labelsep
           \scshape
       #1]\ignorespaces
   }{%
     \endtrivlist\@endpefalse
   }
   \providecommand{\proofname}{Proof}
 \fi
  \theoremstyle{plain}
  \newtheorem{prop}{\protect\propositionname}

\usepackage[usenames,dvipsnames]{xcolor}
\hypersetup{
    colorlinks=true,       
    linkcolor=Maroon,          
    citecolor=OliveGreen,        
    filecolor=magenta,      
    urlcolor=Blue           
}

\makeatother

\providecommand{\lemmaname}{\inputencoding{latin9}Lemma}
\providecommand{\propositionname}{\inputencoding{latin9}Proposition}
\providecommand{\theoremname}{\inputencoding{latin9}Theorem}

\begin{document}
\global\long\def\ket#1{\left|#1\right\rangle }

\global\long\def\Zr{\mathbb{Z}\left[i,1/\sqrt{2}\right]}

\global\long\def\Z{\mathbb{Z}}

\global\long\def\id{\mathbb{I}}

\global\long\def\w{\omega}

\global\long\def\k{\kappa}

\global\long\def\tr{\mathrm{Tr}}

\global\long\def\normf#1{\left\Vert #1\right\Vert _{Fr}}

\global\long\def\norm#1{\left\Vert #1\right\Vert }

\global\long\def\ve{\varepsilon}

\global\long\def\Re{\mathrm{Re}}

\global\long\def\Im{\mathrm{Im}}

\global\long\def\bra#1{\left\langle #1\right|}

\global\long\def\op#1#2{\left|#1\right\rangle \left\langle #2\right|}

\global\long\def\ip#1#2{\left\langle #1|#2\right\rangle }

\title{Synthesis of unitaries with Clifford+T circuits}

\author{Vadym Kliuchnikov}

\email{v.kliuchnikov@gmail.com}

\affiliation{Institute for Quantum Computing, University of Waterloo, Waterloo,
ON, Canada}

\affiliation{David R. Cheriton School of Computer Science, University of Waterloo,
Waterloo, ON, Canada}

\begin{abstract}
We describe a new method for approximating an arbitrary $n$ qubit
unitary with precision $\varepsilon$ using a Clifford and T circuit
with $O(4^{n}n(\log(1/\varepsilon)+n))$ gates. The method is based
on rounding off a unitary to a unitary over the ring $\mathbb{Z}[i,1/\sqrt{2}]$
and employing exact synthesis. We also show that any $n$ qubit unitary
over the ring $\mathbb{Z}[i,1/\sqrt{2}]$ with entries of the form
$(a+b\sqrt{2}+ic+id\sqrt{2})/2^{k}$ can be exactly synthesized using
$O(4^{n}nk)$ Clifford and T gates using two ancillary qubits. This
new exact synthesis algorithm is an improvement over the best known
exact synthesis method by B. Giles and P. Selinger requiring $O(3^{2^{n}}nk)$
elementary gates. 
\end{abstract}
\maketitle

Implementing a unitary operation using a universal gate set is a fundamental
problem in quantum computing. The problem naturally arises when we
want to implement some quantum algorithm on a fault tolerant quantum
computer. Most fault tolerant protocols allow one to implement only
Clifford circuits (those generated by CNOT, Hadamard and Phase gates).
To achieve universal quantum computation one needs to add at least
one non-Clifford gate. One of the common examples is a T:=$\left(\begin{array}{cc}
1 & 0\\
0 & e^{i\pi/4}
\end{array}\right)$ gate. Many unitaries that occur in quantum algorithms cannot be implemented
exactly using the Clifford and T gate set and must be approximated.
However, the study of unitaries that can be implemented exactly using
Cliffords and T has been useful for finding more efficient approximations.
An example is an asymptotically optimal algorithm for approximating
single qubit rotations using Clifford and T gates\cite{noans}. This
algorithm achieves polynomial speed-up over the Solovay-Kitaev algorithm\cite{skalg,kbook}
(which is applicable to generic universal gate sets). The key ingredients
of this algorithm are: an algorithm for synthesizing optimal single
qubit circuits given a unitary that can be implemented exactly\cite{Kliuchnikov2012}
and an efficient round-off procedure that approximates an arbitrary
single qubit unitary by the one that is exactly implementable. We
extend these ideas to the approximate synthesis of unitaries acting
on multiple qubits and show that any $n$ qubit unitary can be approximated
with precision $\varepsilon$ using at most $O\left(4^{n}n\left(\log\left(1/\ve\right)+n\right)\right)$
Clifford and T gates and two ancillae. Our procedure results in a
slightly larger asymptotic number of gates in comparison to first
using an asymptotically optimal decompositions of an $n$ qubit unitary
into $O\left(4^{n}\right)$ single qubit unitaries and CNOT gates\cite{optimalSq,ptrans}
and then approximating each single qubit unitary with Clifford and
T circuit. Both approaches are asymptotically optimal when the number
of qubits is fixed, according to the lower bounds established in \cite{sqans,eda}.
However, our decomposition is more directly related to the structure
of the initial unitary and may be beneficial in two cases: when the
unitary is partially specified, or when most of the entries of the
unitary are over the ring~$\Zr.$ In contrast to decompositions used
in \cite{ptrans,optimalSq,ncbook} our decomposition does not require
taking square roots. We use only addition, subtraction and multiplication
by $1/\sqrt{2}$; these operations preserve elements of $\Zr.$ 

We also improve multiple qubit exact synthesis of unitaries over the
ring $\Zr.$ B. Giles and P. Selinger\cite{mqes} showed that any
$n$ qubit unitary over the ring $\Zr$ with entries of the form $((a+b\sqrt{2})+i(c+d\sqrt{2}))/\sqrt{2}^{\kappa}$
can be synthesized exactly using at most one ancilla. However, their
synthesis method requires $O\left(3^{2^{n}}n\kappa\right)$ gates,
which is far from information theoretic bounds. In this paper we propose
a synthesis method that uses two ancillae and requires $O\left(4^{n}n\kappa\right)$
gates. The main results of the paper are summarized in the following
theorems: 
\begin{thm}
\label{thm:approx}Any $n$ qubit unitary can be approximated within
Frobenius distance $\varepsilon$ using $O\left(4^{n}n\left(C\log\left(1/\ve\right)+n\right)\right)$
Clifford and T gates and at most two ancillae.
\end{thm}

\begin{thm}
\label{thm:exact}Any $n$ qubit unitary with entries of the form
$((a+b\sqrt{2})+i(c+d\sqrt{2}))/\sqrt{2}^{\k}$ can be exactly implemented
using $O\left(4^{n}n\k\right)$ Clifford and T gates using at most
two ancillae. 
\end{thm}
To prove both results we use a variant of the Householder decomposition
which expresses a matrix as a product of reflection operators and
a diagonal unitary matrix. In our case the diagonal unitary matrix
is always the identity. We define a reflection operator constructed
from a unit vector $\ket{\psi}$ as $R_{\ket{\psi}}=\id-2\op{\psi}{\psi}$.
Our structure-preserving decomposition is given by the following lemma: 
\begin{lem}
\label{lem:refl-decomp}Let $U$ be a unitary acting on $n$ qubits
and let $\left\{ u_{1},\ldots,u_{2^{n}}\right\} $ be the columns
of $U.$ The unitary $U$ can be simulated using the unitary
\[
U'=\op 01\otimes U+\op 10\otimes U^{\dagger}.
\]
Unitary $U'$ is a product of reflection operators constructed from
the family of unit vectors 
\[
\ket{w_{j}^{\pm}}=\left(\ket 1\otimes\ket j\pm\ket 0\otimes\ket{u_{j}}\right)/\sqrt{2},\mbox{ for }j=1,\ldots,2^{n}.
\]
\end{lem}
\begin{proof}
By direct calculation we check that $U'$ maps $\ket 1\otimes\ket{\phi}$
to $\ket 0\otimes U\ket{\phi}$ for any $n$ qubit state $\ket{\phi}.$
Next we observe that $\ket{w_{j}^{\pm}}$ are eigenvectors of $U'$
with eigenvalues $\pm1.$ Defining $P_{j}^{\pm}$ to be projectors
on subspaces spanned by $\ket{w_{j}^{\pm}}$ and using the spectral
theorem we express $U'$ as $\sum_{j=1}^{2^{n}}\left(P_{j}^{+}-P_{j}^{-}\right).$
Since $\sum_{j=1}^{2^{n}}\left(P_{j}^{+}+P_{j}^{-}\right)$ is the
identity operator and all projectors $P_{j}^{\pm}$ are orthogonal
we can write: 
\[
U'=I-2\sum_{j=1}^{2^{n}}P_{j}^{-}=\prod_{j=1}^{2^{n}}\left(I-2P_{j}^{-}\right).
\]
The right hand side is a product of $2^{n}$ reflection operators,
as required. 
\end{proof}
It is not difficult to see that if $\ket{u_{j}}$ has coordinates
in the computational basis over the ring $\Zr$ then the unit vector
$\ket{w_{j}^{\pm}}$ also does. This the reason why we call our decomposition
structure-preserving. Exactly synthesizing a reflection operator is
not more difficult than exactly preparing corresponding unit vector: 
\begin{lem}
\label{lem:refl-synth}Any reflection operator $R_{\ket{\phi}}$
where $\ket{\phi}$ has coordinates in the computational basis of
the form $((a+b\sqrt{2})+i(c+d\sqrt{2}))/\sqrt{2}^{\k}$ can be implemented
using $O\left(2^{n}n\k\right)$ Clifford and T gates and at most one
ancillae. \end{lem}
\begin{proof}[Proof of \lemref{refl-synth}]
 Note that $UR_{\ket 0}U^{\dagger}=R_{U\ket 0}$. Therefore, to implement
the reflection operator with corresponding unit vector $\ket{\phi}$
it is enough to find a $U$ that prepares $\ket{\phi}$ starting from
$\ket 0.$ The column lemma \cite{mqes} provides a construction for
$U$ requiring $O\left(2^{n}n\k\right)$ Clifford and T gates and
one ancilla. Unitary $R_{\ket 0}$ is a multiple controlled $Z$ operator
and can be implemented with $O\left(n\right)$ gates and one ancilla,
for example using the construction from \cite{ncbook}. We conclude
that we need $O\left(2^{n}n\k\right)$ Clifford and T gates in total
to implement $R_{\ket{\phi}}.$ 
\end{proof}
Now we have all results required to proof Theorem \ref{thm:exact}:
\begin{proof}[Proof of \thmref{exact}]
The construction from \lemref{refl-decomp} allows one to simulate
an $n$ qubit unitary using one ancilla and $2^{n}$ reflection operators.
The unit vectors corresponding to each reflection operator have coordinates
of the form $((a+b\sqrt{2})+i(c+d\sqrt{2}))/\sqrt{2}^{\k+1}$ in the
computational basis. According to \lemref{refl-synth} these reflection
operators can be implemented using one ancilla and $O\left(2^{n}n\k\right)$
Clifford and T gates. Therefore we need $O\left(4^{n}n\k\right)$
Clifford and T gates and at most two ancillae to implement the unitary
exactly. 
\end{proof}
To show the approximation result we use the decomposition above and
then approximate each reflection operator separately. First we note
the following relation between approximating reflection operators
and their corresponding unit vectors:
\begin{prop}
\label{prop:refl-dist}The distance induced by the Frobenius norm
between two reflection operators is bounded by the Euclidean distance
between corresponding unit vectors as: 
\[
\normf{R_{\ket{\psi}}-R_{\ket{\phi}}}\le2\sqrt{2}\norm{\ket{\psi}-\ket{\phi}}.
\]

\end{prop}
To prove the proposition it is enough to use the definition of the
Frobenius norm $\normf A^{2}=\tr\left(AA^{\dagger}\right),$ express
$\normf{R_{\ket{\psi}}-R_{\ket{\phi}}}^{2}$ in terms of $\left|\ip{\phi}{\psi}\right|^{2},$
use that $\Re\ip{\phi}{\psi}\le\left|\ip{\phi}{\psi}\right|$ and
express $\Re\ip{\phi}{\psi}$ in terms of $\norm{\ket{\psi}-\ket{\phi}}^{2}.$
Next we show how to approximate arbitrary reflection operator by at
most two reflection operators with corresponding unit vectors over
the ring $\Zr$.

\begin{lem}
\label{lem:approx}Any $n$ qubit reflection operator $R_{\ket{\psi}}$
can be approximated within Frobenius distance $\varepsilon$ by the
product of two reflection operators, such that coordinates in the
computational basis of corresponding unit vectors have the form $(a+b\sqrt{2}+ci+di\sqrt{2})/2^{m\left(\ve\right)}$
where $m\left(\ve\right)=\left\lceil n/2\right\rceil +O\left(\log\left(1/\ve\right)\right)$
and using at most one ancilla. If unit vector $\ket{\psi}$ has at
least two coordinates in the computational basis equal to zero it
is sufficient to use one reflection operator and no ancilla is required.\end{lem}
\begin{proof}[Proof of \lemref{approx}]
First we construct the approximating unitary in a special case where
no ancilla are required. Consider the reflection operator $R_{\ket{\psi}}$.
Let $\left\{ \alpha_{k}\right\} $ be the coordinates of $\ket{\psi}$
in computational basis. First we consider the case when $\ket{\psi}$
has at least two zero entries, say $\alpha_{j}$ and $\alpha_{l}$.
We use an idea similar to \cite{sqans} and define the approximating
unitary as a reflection operator $R_{\ket{\phi}}$ corresponding to
the vector 
\begin{eqnarray*}
\ket{\phi} & = & \frac{a+bi}{2^{m}}\ket j+\frac{c+di}{2^{m}}\ket l+\sum_{k=1,\, k\ne j,l}^{2^{n}}\beta_{k}\ket k,\\
 &  & \beta_{k}=\frac{\left\lfloor 2^{m}\Re\alpha_{k}\right\rfloor +i\left\lfloor 2^{m}\Im\alpha_{k}\right\rfloor }{2^{m}},a,b,c,d\in\Z.
\end{eqnarray*}
The norm of $\ket{\phi}$ must be equal to $1$, therefore: 
\begin{equation}
a^{2}+b^{2}+c^{2}+d^{2}=4^{m}\left(1-\sum_{k=1,\, k\ne j,l}^{2^{n}}\left|\beta_{k}^{2}\right|\right).\label{eq:4sq}
\end{equation}
The Diophantine equation above always has a solution according to
Lagrange's four-square theorem and it can be efficiently found using
a probabilistic algorithm\cite{Rabin1986} that has runtime polynomial
in number of bits required to write the right hand side of equation
(\ref{eq:4sq}).

By Proposition \propref{refl-dist}, to estimate the distance between
the reflection operator and its approximation it is enough to estimate
the square of the distance between $\ket{\psi}$ and $\ket{\phi}.$
We approximated each non-zero entry with precision $2^{-m}\sqrt{2}$,
therefore 
\[
\norm{\ket{\psi}-\ket{\phi}}^{2}\le2^{n}\left(2^{-m}\sqrt{2}\right)^{2}+4^{-m}\left(a^{2}+b^{2}+c^{2}+d^{2}\right).
\]
The second summand of the right hand side of the inequality above
can be estimated as: 
\begin{eqnarray*}
1-\sum_{k=1,\, k\ne j,l}^{2^{n}}\left|\beta_{k}^{2}\right| & = & \sum_{k=1,\, k\ne j,l}^{2^{n}}\left(\left|\alpha_{k}^{2}\right|-\left|\beta_{k}^{2}\right|\right)\\
 & \le & \sum_{k=1,\, k\ne j,l}^{2^{n}}\left|\left|\alpha_{k}\right|-\left|\beta_{k}\right|\right|\left(\left|\alpha_{k}\right|+\left|\beta_{k}\right|\right)\\
 & \le & 2^{-m}\sqrt{2}\sum_{k=1,\, k\ne j,l}^{2^{n}}\left(\left|\alpha_{k}\right|+\left|\beta_{k}\right|\right)\\
 & \le & 2^{-m}\sqrt{2}\left(2\cdot2^{n/2}\right).
\end{eqnarray*}
We used the Cauchy–Schwarz inequality to estimate sums involving $\left|\alpha_{k}\right|$
and $\left|\beta_{k}\right|.$ For example: 
\[
\sum_{k=1,\, k\ne j,l}^{2^{n}}\left|\alpha_{k}\right|\le\sqrt{2^{n}}\sqrt{\sum_{k=1,\, k\ne j,l}^{2^{n}}\left|\alpha_{k}\right|^{2}}=2^{n/2}.
\]
In summary we get: 
\[
\norm{\ket{\psi}-\ket{\phi}}^{2}\le2\cdot2^{n-2m}+2\sqrt{2}\cdot2^{\left(n/2-m\right)}.
\]
By choosing $m=\left\lceil n/2+\log_{2}\left(1/\ve^{2}\right)\right\rceil +5$
and using Proposition \propref{refl-dist} we get $\normf{R_{\ket{\psi}}-R_{\ket{\phi}}}\le\ve$
when $\ve\le1$.

In the case when all entries of $\ket{\psi}$ in the computational
basis are non-zero, we add an ancilla and express the reflection around
$\ket{\psi}$ using two reflection operators with unit vectors that
can be approximated without using ancilla: 
\begin{eqnarray*}
\id_{1}\otimes R_{\ket{\psi}} & = & \id_{1}\otimes\id_{n}-2\id_{1}\otimes\op{\psi}{\psi}\\
 & = & \id_{1}\otimes\id_{n}-2\left(\op 00+\op 11\right)\otimes\op{\psi}{\psi}\\
 & = & R_{\ket{0\psi}}R_{\ket{1\psi}}.
\end{eqnarray*}

\end{proof}
Now we have all tools needed to prove Theorem \ref{thm:approx}: 
\begin{proof}[Proof of \thmref{approx}]
We use construction from \lemref{refl-decomp} to simulate an $n$
qubit unitary using one ancilla and $2^{n}$ reflection operators.
The unit vectors corresponding to each reflection operator have at
least two zero coordinates in the computational basis, therefore we
can use \lemref{approx} to approximate each reflection with a reflection
operator without using ancillae with precision $2^{-n}\ve.$ Unit
vectors of each approximating reflection have entries of the form
$(a+b\sqrt{2}+ci+di\sqrt{2})/2^{m\left(n,\ve\right)}$ for $m\left(n,\ve\right)=\left\lceil 5n/2\right\rceil +O\left(\log\left(1/\ve\right)\right).$
Each reflection operator requires one ancillae and $O\left(2^{n}n\cdot m\left(n,\ve\right)\right)$
Clifford and T gates according to \lemref{refl-synth}. Therefore,
in total we need two ancillae and $O\left(4^{n}n\left(\log\left(1/\ve\right)+n\right)\right)$
Clifford and T gates to approximate the unitary within Frobenius distance
$\varepsilon.$ 
\end{proof}

Our improved method for multi qubit exact synthesis outputs circuits
with an exponential number gates as a function of the number of qubits.
This improves previously known method\cite{mqes} which requires doubly
exponential number of gates. The further improvements over our result
may be possible: for example, removing the factor of $n$ from an
expression $O\left(4^{n}n\left(\log\left(1/\ve\right)+n\right)\right).$ 
\begin{acknowledgments}
We wish to thank David Gosset, Michele Mosca, Martin Roetteler and
Peter Selinger for helpful discussions. 

The author is supported in part by the Intelligence Advanced Research
Projects Activity (IARPA) via Department of Interior National Business
Center Contract number DllPC20l66. The U.S. Government is authorized
to reproduce and distribute reprints for Governmental purposes notwithstanding
any copyright annotation thereon. Disclaimer: The views and conclusions
contained herein are those of the authors and should not be interpreted
as necessarily representing the official policies or endorsements,
either expressed or implied, of IARPA, DoI/NBC or the U.S. Government.
\end{acknowledgments}
\bibliographystyle{plainurl}

\end{document}